\documentclass{article}

\usepackage[final,main]{neurips_2026}

\usepackage[utf8]{inputenc}
\usepackage[T1]{fontenc}
\usepackage{microtype}
\usepackage{amsmath,amssymb,amsfonts,amsthm,mathtools}
\usepackage{booktabs}
\usepackage{url}
\usepackage{xcolor}
\usepackage{hyperref}
\usepackage{enumitem}
\usepackage{algorithm}
\usepackage{algorithmic}

\theoremstyle{plain}
\newtheorem{theorem}{Theorem}[section]
\newtheorem{proposition}[theorem]{Proposition}
\newtheorem{lemma}[theorem]{Lemma}
\newtheorem{corollary}[theorem]{Corollary}
\theoremstyle{definition}
\newtheorem{definition}[theorem]{Definition}
\newtheorem{assumption}[theorem]{Assumption}
\theoremstyle{remark}

\title{Global Sketch-Based Watermarking for Diffusion Language Models}

\author{%
  Daniel Zhao \\
  Harvard University
}

\begin{document}

\maketitle
\begin{abstract}
Watermarking methods for language models have been studied extensively in the autoregressive setting, where tokens are generated sequentially. These works largely focus on local-context schemes that perturb the next token's distribution as a function of its preceding tokens. In diffusion language models, distributions over many unresolved positions are jointly sampled, allowing additive statistics of the entire sequence to be tractable during generation. We propose a watermark for masked diffusion language models that controls a global, vector-valued sketch representation of the text. Compared to context-dependent watermarking, the sketch formulation decouples detection from the local contexts seen during generation, resulting in an order-agnostic statistic and a watermarking rule which does not manifest as a simple token bias. We analyze the distortion, soundness, and robustness properties of the method. 
\end{abstract}

\section{Introduction}

The rapid adoption of large language models has made content attribution an increasingly urgent problem. To mitigate abuse and compounding data pollution, we require reliable methods to distinguish human-written text from model-generated text \citep{Pudasaini2025AIGeneratedPlagiarismSurvey,Paustian2024}. Watermarking provides a statistical approach to this problem by embedding detectable signatures into the text that are invisible to ordinary readers.

A large body of work studies watermarking for autoregressive (AR) language models, where tokens are conditionally generated left to right. The largest family of watermarking schemes in this category biases the next-token distribution based on a function of the current prefix, typically by partitioning the vocabulary into favored and disfavored sets. Each token becomes fixed once sampled. This ties the watermark to the left-right filtration generated by autoregression which, while natural, makes it difficult to separate the signal from its sampling path.

By contrast, diffusion-based language models (DLMs) generate tokens by denoising corrupted sequences of text simultaneously \citep{austin2021d3pm,Hoogeboom2021Multinomial,li2022diffusionlm}. In discrete masked variants, the model exposes marginal distributions over all undetermined token positions. A watermark can then draw on the entire sequence state rather than a realized prefix, opening up global control strategies that are unavailable, or computationally infeasible, for AR models.
 
An effective watermarking scheme seeks to maximize four complementary properties.
\begin{itemize}
\item \textbf{Correctness.} Under reasonable entropy assumptions, the watermark can be detected with high power for a fixed false positive rate.
\item \textbf{Quality.} The observable behavior of the watermarked model remains essentially unchanged from the base model. This typically means minimizing the divergence between the watermarked and base model distributions, known as distortion. Distortion-free schemes preserve the base distribution exactly. 
\item \textbf{Security.} For an unkeyed user, the effect of the watermark should not be easy to learn or imitate. A related notion to distortion-freeness, but strictly stronger, is \emph{undetectability}, which requires that no unkeyed observer can efficiently distinguish watermarked outputs from base-model outputs \citep{christ2024undetectable}.
\item \textbf{Robustness.} The watermark persists through post-generation edits. A robust scheme makes it expensive for an attacker to remove a watermark, requiring many token changes or costly paraphrasing to degrade detection power.
\end{itemize}

In this preliminary work, we introduce a sketch-based watermark for diffusion language models. The sketch is an additive vector statistic of the output text whose conditional expectation is tractable from the marginals during denoising. We bias those marginals by following the gradient of a detector score evaluated at the current expected sketch. Our design leverages the native geometry of DLMs for order-agnostic, sequence-level control. We prove a chain rule for KL distortion, a key-randomized soundness bound, and bounds on the edit sensitivity of the statistic.

\section{Preliminaries}

\subsection{Masked diffusion language models}
\label{sec:prelim:dlm}

Let $\mathcal{V}$ be a finite vocabulary.  Given a prompt $c$, a diffusion language model (DLM) defines a reverse-time Markov chain over a state space. In this work, we focus on the masked variant with the factorization assumption.

\begin{definition}[Masked diffusion sampler]
\label{def:dlm}
For a prompt $c$ and length $N$, a masked diffusion language model $\mathcal M$
defines a reverse-time Markov process on $(\mathcal V\cup\{\texttt{[MASK]}\})^N$, starting from an initial state $S_T\sim \pi_T(\cdot\mid c)$. At each step $t=T,\dots,1$, let $M_t:=\{i:(s_t)_i=\texttt{[MASK]}\}$ be the set of masked positions. The marginals $p_i^{(t)}(\cdot\mid s_t,c)\in\Delta(\mathcal V)$ are known for each $i\in M_t$. The model proceeds by sampling the masked positions from these marginals conditionally independent of each other, leaving the unmasked coordinates unchanged, and then applying a remasking rule to obtain $S_{t-1}$. We write the final output as $Y:=S_0\in\mathcal V^N$ and $Y\sim \mathcal M^N(c).$
\end{definition}

\begin{definition}[Watermarking scheme]
\label{def:wm_scheme}
A watermarking scheme for $\mathcal M$ is a pair
\[
(\widetilde{\mathcal M}, \mathsf{Detect}),
\]
where, for each key $k$, $\widetilde{\mathcal M}_k$ is a keyed modification of $\mathcal M$, and
\[
\mathsf{Detect}_k:\mathcal V^N\to\{0,1\}
\]
is the corresponding detector. For a prompt $c$ and length $N$, we write $
\widetilde Y \sim \widetilde{\mathcal M}_k^N(c).$
\end{definition}

\subsection{Text sketches}
A \emph{text sketch} is a map \(H:\mathcal V^*\to\mathcal E\), where \(\mathcal E\) is a finite-dimensional real vector space. Sketches are approximate summaries that retain enough information about a text's features to support operations such as similarity search, clustering, duplication detection, and indexing \citep{charikar2002similarity,weinberger2009feature,manning2008ir}. In our watermark, we use a form of Count-Sketch, a randomized linear method widely used in streaming applications. 

\begin{definition}[Count-Sketch]
\label{def:count_sketch}
Fix integers $d,w\ge 1$. For each row $r\in[d]$, instantiate
\[
h_r:\mathcal V\to [w],
\quad
\sigma_r:\mathcal V\to\{-1,+1\},
\]
a hash function and a sign function, respectively. Define the per-token feature map $\phi:\mathcal V\to \mathbb R^{d\times w}$ by
\[
\phi(v)_{r,b}:=\sigma_r(v)\,\mathbf 1[h_r(v)=b],
\quad
r\in[d],\ b\in[w].
\]
Then, for a text $y=(y_1,\dots,y_N)\in\mathcal V^*$, its sketch is
\begin{equation}
H(y)
:=
\sum_{j=1}^N \phi(y_j)
\in \mathbb R^{d\times w}.
\label{eq:countsketch-H}
\end{equation}
\end{definition}
Note that $H$ is additive over tokens, a property used in the embedding and robustness analysis.

\section{Related Work}

\subsection{Autoregressive watermarking}

Broadly, existing AR watermarking schemes use one of three mechanisms to inject signal: biasing the next-token distribution, replacing sampling randomness with pseudorandom functions, or rejection sampling, which resamples candidates until a constraint is satisfied. Kirchenbauer et al.\ introduced the canonical red-green biasing approach, which uses a hash of the preceding context to partition the vocabulary into two sets, and then boosts the logits of the "green'' set. Variants in this family optimize for robustness through globally fixed partitions \citet{zhao2023provable} or semantic context hashing \citep{guo2024contextaware}. Others avoid biasing logits directly, but use sampling choices to change the distribution. SynthID-Text, for example, performs tournament sampling over candidates to favor tokens that have large secret score values \citep{dathathri2024scalable}.

The second family of watermarks preserves the marginal distribution of each token, and instead replaces the sampler's randomness with a pseudorandom function. The approach taken by Kuditipudi et al.\ uses a long, finite, shared pseudorandom sequence to produce distortion-free samples \citep{kuditipudi2024robust}. Detection becomes a robust alignment problem. Christ et al.\ design undetectable watermarks by ensuring the pseudorandom sequence is not exactly reused on multiple responses \citep{christ2024undetectable}. Aaronson’s Gumbel-style construction is often cited as a precursor, though common presentations of the scheme are not distortion-free \citep{aaronson2023watermarking}.

Rejection-sampling approaches repeatedly propose candidate tokens or sequences of tokens from the base model, and only accept when a constraint is satisfied. This allows the watermark signal to live in more robust, higher-order features than individual tokens, such as semantic meaning \citep{hou2023semstamp,dabiriaghdam2025simmark,huo2025pmark}, and further supports public-key cryptography \citet{fairoze2023publiclydetectable}.

\subsection{Limitations of AR Watermarking}

\textbf{Single-pass control bottleneck.}
In most instances, AR watermarks accumulate evidence through a sequence of realized, prefix-conditioned decisions. A token's contribution to the detector is fixed after it is sampled. Furthermore, the scheme does not have access to a joint source of uncertainty over the whole sequence. Therefore—with the notable exception of some distortion-free schemes—the watermark signal must be carried by local perturbations, limiting the degrees of freedom available for embedding while precluding direct control of sequence-level statistics. Rejection-sampling approaches can partially mitigate the one-pass constraint, but only by expending extra compute via an outer loop \citep{hou2023semstamp, fairoze2023publiclydetectable}.

\textbf{Security of local hashing schemes.}
A second limitation of local hashing schemes is security. Since these watermarks induce repeatable, context-conditioned regularities per token, they have been shown to be detectable by unkeyed users \citep{gloaguen2025blackbox}; for the red-green scheme, an adaptive adversary can estimate the watermark effect and spoof it without even having access to the base model or a detector oracle \citep{jovanovic2024watermarkstealing, gu2024learnability}. 

\textbf{Synchronization robustness.}
Most AR watermark detectors rely on the sequential token ordering to synchronize their internal states with the observed text, making them vulnerable to insertion, deletion, and reordering attacks. Targeted attacks can quickly degrade power unless the scheme adds explicit defenses \citep{jovanovic2024watermarkstealing}. Unfortunately, these defense mechanisms can also worsen the learnability problem \citep{gu2024learnability}.

\subsection{Diffusion language models}

Diffusion language models (DLMs) generate text by denoising over partially corrupted sequences, updating many token positions at each time step. Much of the early work adapted diffusion to language by mapping tokens into a continuous latent space, though discrete variants that perform diffusion directly in the token space are becoming more prevalent. The forward process corrupts sequences via masking or resampling, and the model learns to reverse that process over categorical states \citep{Hoogeboom2021Multinomial,austin2021d3pm,nie2025llada}. While not as mature as their autoregressive counterparts, DLMs have steadily narrowed the empirical gap to strong autoregressive baselines and may offer faster inference \citep{wang2025d2f}, greater sampling control \citep{schiff2025path}, and global coherence \citep{ma2026diffusionindiffusion}.

Watermarking for DLMs is only beginning to be explored. Both distortion-free methods and the red-green watermark have been studied for masked DLMs \citep{gloaguen2025dlmwatermark, bagchi2025ddlm}. The red-green scheme of Gloaguen et al.\ maximizes the expected frequency of green list tokens over the marginals. They also introduce two interesting diffusion-native extensions: a predictive bias term which makes future tokens more likely to be green, and entropy-based remasking that updates high-uncertainty positions first.

\section{Methodology}
\label{sec:method}

Our watermark is a sequence-level control rule for masked DLMs. The main idea is to represent the text by a sketch function $h(y) \in \mathbb{R}^D$ and guide generation to align $h$ with a target direction $u$. Under a key $k$, we detect with a regularized dot-product score
\begin{equation}
S_{k}(y)
:=
2 \langle u_k,h_k(y)\rangle
- \lambda \lVert h_k(y)\rVert_2^2,
\label{eq:penalized-score}
\end{equation}
where $u_k\in\{\pm1\}^D$ and $\lambda>0$ penalizes the sketch norm. Thus the score rewards agreement with $u_k$ but discourages the sampler from increasing the sketch norm indiscriminately. In fact, completing the square gives
\begin{equation}
S_{k}(y)
=
-\lambda
\left\|
h_k(y)-\frac{1}{\lambda}u_k
\right\|_2^2
+
\frac{1}{\lambda}\|u_k\|_2^2,
\label{eq:penalized-template-equivalence}
\end{equation}
which is maximized when $h_k(y)$ is close, in Euclidean distance, to the scaled target $\frac{u_k}{\lambda}.$ We discuss the choice of $\lambda$ below. 

To embed the watermark, during denoising the DLM marginals are biased along the direction of the detector gradient at the current expected sketch. This is possible without requiring the joint distribution due to the token-additivity of the sketch, which allows us to compute its expectation under the conditional marginals.

\subsection{Sketch design}
\label{sec:method:sketch}

We use a Count-Sketch style function from Definition~\ref{def:count_sketch} for $h_k$. The secret key $k$ is a pair of independent values $k=(k_{\mathrm{sk}},k_{\mathrm{dir}})$, with the hash and sign functions derived from a sketch key $k_{\mathrm{sk}}$, and the Rademacher direction $u_k$ derived from $k_{\text{dir}}$. This construction separates the coordinate system induced by the sketch from the direction tested by the detector, both of which are private. Write the token feature map as
\[
\phi_k:\mathcal V\to\mathbb R^{d\times w},
\quad
h_k(y):=\mathrm{vec}\left(\frac{1}{\sqrt{N}}\sum_{j=1}^N \phi_k(y_j)\right),
\]
such that the sketch dimension is $D = dw$. Each token activates one bucket per row. The signs are chosen such that token density under the base law $\pi_{\text{tok}}$ is approximately balanced for every bucket, that is
\begin{equation}
\sum_{v:\,h_r(v)=b}\pi_{\text{tok}}(v)\,\sigma_r(v)\approx 0,
\label{eq:balanced-signs-method}
\end{equation}
which suppresses the most obvious first-order token footprint. Under this condition, each coordinate of \(h_k(Y)\) has approximately zero mean and standard deviation of order \(w^{-1/2}\). We set the penalty scalar $\lambda=\gamma\sqrt w$ for a tunable constant $\gamma > 0$ such that the score feedback term \(\lambda \widehat h_{k,t}\) has the same scale as \(u_k\). Equivalently, the induced sketch target \(u_k/\lambda\) lies on the normalized scale of attainable sketches:
\[
\left\|u_k/\lambda\right\|_2 \asymp \sqrt d
\asymp \left\|h_k(Y)\right\|_2 .
\]
Here, \(\gamma\) controls the strength of the scale correction. As \(\gamma\) decreases, the rule approaches a fixed red-green watermark similar to \cite{zhao2023provable}; as \(\gamma\) increases, the residual-gradient direction becomes more strongly determined by the current denoising state.

\subsection{Score-gradient embedding}
\label{sec:method:embed}

Let $S_t$ be the masked DLM state at reverse time $t$. Define
\[
F_t:=\{i:(S_t)_i\ne{\tt[MASK]}\},\quad
M_t:=\{i:(S_t)_i={\tt[MASK]}\}.
\]
For each $i\in M_t$, the DLM exposes a marginal distribution $p_i^{(t)}(\cdot\mid S_t,c)\in\Delta(\mathcal V).$ Since the sketch is additive, its conditional expectation under the current masked state is tractable:
\begin{equation}
\begin{gathered}
\widehat H_{k,t}
:=
\sum_{i\in F_t}\phi_k((S_t)_i)
+
\sum_{i\in M_t}
\sum_{v\in\mathcal V}
p_i^{(t)}(v\mid S_t,c)\,\phi_k(v),
\\
\widehat h_{k,t}
:=
\operatorname{vec}\left(\frac{1}{\sqrt N}\widehat H_{k,t}\right).
\end{gathered}
\label{eq:expected-sketch-method}
\end{equation}
Here, $\widehat h_{k,t}$ acts as a one-step surrogate for the final sketch. We aim to nudge it towards $u_k$ with each step. The gradient of the detector score with respect to a generic $h\in\mathbb R^D$ is
\[
\nabla_h S_{k}(h)=2(u_k-\lambda h).
\]
Let the residual-gradient direction be $r_{k,t}:=u_k-\lambda\widehat h_{k,t}.$ For a candidate token $v$, the dot product
\begin{equation}
a_{k,t}(v)
:=
\frac{1}{\sqrt N}\left\langle r_{k,t},
\operatorname{vec}(\phi_k(v))
\right\rangle
\label{eq:local-gradient-score}
\end{equation}
specifies the strength and direction of the watermark bias for token $v$. Since \(r_{k,t}\) depends on the current denoising state through \(\widehat h_{k,t}\), the preference is state-dependent: a token preferred at one step may be unfavorable on another. We implement this bias in late denoising steps by exponentially tilting each masked marginal with strength \(\eta_t>0\):
\begin{equation}
\widetilde{p}_{i}^{(t)}(v)
:=
\frac{
p_i^{(t)}(v\mid S_t,c)
\exp\{\eta_t a_{k,t}(v)\}
}{
\sum_{z\in\mathcal V}
p_i^{(t)}(z\mid S_t,c)
\exp\{\eta_t a_{k,t}(z)\}
}.
\label{eq:local-exponential-tilt}
\end{equation}
The bias slope need not, in principle, be shared across token positions. One might expect to allocate a larger $\eta_t$ to positions for which the local score has higher variance under the base marginal, since such positions permit a larger change in expected score per unit KL cost. However, Proposition ~\ref{prop:minimum-kl-local-gain} shows that a position-uniform choice of \(\eta_t\) is locally KL-optimal.

\begin{algorithm}[t]
\small
\caption{Embedding in a masked DLM}
\label{alg:embed-method}
\begin{algorithmic}[1]
\STATE \textbf{Input:} prompt $c$, key $k$, sketch map $\phi_k$, direction $u_k$, penalty $\lambda$, schedule $\{\eta_t\}$
\STATE Sample $S_T\sim \pi_T(\cdot\mid c)$
\FOR{$t=T,T-1,\dots,1$}
  \STATE Obtain DLM marginals $p_i^{(t)}(\cdot\mid S_t,c)$ for $i\in M_t$
  \IF{$\eta_t>0$}
    \STATE Compute $\widehat h_{k,t}$ from~\eqref{eq:expected-sketch-method}
    \STATE $r_{k,t}\leftarrow u_k-\lambda\widehat h_{k,t}$
    \FOR{each $i\in M_t$}
      \STATE Replace $p_i^{(t)}$ with $\widetilde p_i^{(t)}$ from~\eqref{eq:local-exponential-tilt}
    \ENDFOR
  \ENDIF
  \STATE Apply the DLM update and base remasking rule using the current marginals
\ENDFOR
\STATE \textbf{return} $Y=S_0$
\end{algorithmic}
\end{algorithm}

\subsection{Detection}
\label{sec:method:detect}

Given a candidate text $y=(y_1,\dots,y_N)$ and key $k$, the detector computes $S_{k}(y)$ and applies a threshold:
\begin{equation}
\mathsf{Detect}_k(y)
:=
\mathbf 1\{S_{k}(y)\ge \tau_N\}.
\label{eq:detect-rule}
\end{equation}
The threshold $\tau_N$ is calibrated under the base model, conditional on the output length, to control the false positive rate. The bound in Theorem~\ref{thm:key-randomized-soundness} gives a sufficient analytic threshold; in practice, for each length $N$ or length bin, we may draw unwatermarked samples $Y^{(1)},\ldots,Y^{(B)}\sim \mathcal M^N(c)$ over the same prompt distribution used at deployment and set $\tau_N$ to the empirical $(1-\alpha)$-quantile of $S_k(Y^{(b)})$. This absorbs average prompt effects and residual imbalances in the sketch coordinates.

\begin{algorithm}[t]
\small
\caption{Detection}
\label{alg:detect-method}
\begin{algorithmic}[1]
\STATE \textbf{Input:} text $y$, key $k$, sketch $h_k$, direction $u_k$, penalty $\lambda$, threshold $\tau_N$
\IF{$S_{k}(y)\ge \tau_N$}
  \STATE \textbf{return} watermarked
\ELSE
  \STATE \textbf{return} not watermarked
\ENDIF
\end{algorithmic}
\end{algorithm}

\section{Theoretical Results}
\label{sec:theory}

We prove three basic properties of the watermark. First, KL distortion decomposes along the denoising path. Second, the exponential tilt in Equation \eqref{eq:local-exponential-tilt} has the usual Fisher form: under a small-bias approximation, a common \(\eta_t\) is locally optimal under first, and high-variance positions automatically carry more signal. Third, the detector has a key-randomized soundness bound and degrades linearly under token edits.

\subsection{KL chain rule and local Fisher allocation}
\label{sec:theory:kl-fisher}

The first result is a path-space accounting identity. It is a standard result for Markov chains, specialized to the diffusion setting, and is independent of the particular sketch.

\begin{theorem}
\label{thm:kl_chain_rule}
Let \(P\) and \(\widetilde P\) be the base and watermarked laws on trajectories \((S_T,\ldots,S_0)\), with common initial law \(\pi_T(\cdot\mid c)\) and reverse kernels
\(Q_t(\cdot\mid S_t,c)\) and \(\widetilde Q_t(\cdot\mid S_t,c)\). Assume
\(\widetilde Q_t(\cdot\mid s,c)\ll Q_t(\cdot\mid s,c)\) for \(\widetilde P\)-a.e. \(s\). Then
\begin{equation}
D_{\mathrm{KL}}(\widetilde P\|P)
=
\sum_{t=1}^{T}
\mathbb E_{\widetilde P}\!\left[
D_{\mathrm{KL}}\!\left(
\widetilde Q_t(\cdot\mid S_t,c)
\,\middle\|\,
Q_t(\cdot\mid S_t,c)
\right)
\right].
\label{eq:kl-chain-main}
\end{equation}
If \(P_Y\) and \(\widetilde P_Y\) denote the induced laws of the final text \(Y=S_0\), then
\begin{equation}
D_{\mathrm{KL}}(\widetilde P_Y\|P_Y)
\le
D_{\mathrm{KL}}(\widetilde P\|P).
\label{eq:kl-data-processing-main}
\end{equation}
\end{theorem}

\begin{proof}
The common initial law cancels in the trajectory likelihood ratio:
\[
\log\frac{d\widetilde P}{dP}(S_{0:T})
=
\sum_{t=1}^T
\log
\frac{d\widetilde Q_t(\cdot\mid S_t,c)}
     {dQ_t(\cdot\mid S_t,c)}(S_{t-1}).
\]
Taking expectation under \(\widetilde P\) and conditioning on \(S_t\) gives
\eqref{eq:kl-chain-main}. Since \(Y=S_0\) is a measurable function of the trajectory,
\eqref{eq:kl-data-processing-main} follows from the monotonicity of KL divergence under measurable transformations
\citep{cover2006elements}.
\end{proof}

If the remasking rule is identical to the unwatermarked sampler—which is not guaranteed in our scheme—and masked positions are sampled conditionally independently (i.e. the "factorization assumption") then
\begin{equation}
D_{\mathrm{KL}}\!\left(
\widetilde Q_t(\cdot\mid s_t,c)
\,\middle\|\,
Q_t(\cdot\mid s_t,c)
\right)
\le
\sum_{i\in M_t}
D_{\mathrm{KL}}\!\left(
\widetilde p_i^{(t)}(\cdot\mid s_t,c)
\,\middle\|\,
p_i^{(t)}(\cdot\mid s_t,c)
\right).
\label{eq:transition-kl-position-bound}
\end{equation}
Equality holds for tokens drawn before remasking.

We next show that the bias allocation rule is approximately KL-optimal per time step. For a step \(t\) and \(i\in M_t\), write \(p_i^{(t)}(\cdot)=p_i^{(t)}(\cdot\mid s_t,c)\), and define
\[
\mu_{t,i}:=\mathbb E_{p_i^{(t)}}[a_{k,t}(v)],
\qquad
A_{t,i}(\eta)
:=
\log
\mathbb E_{p_i^{(t)}}[\exp\{\eta a_{k,t}(v)\}].
\]
Let \(\eta\ge 0\) be a common slope, and
\begin{equation}
p_{i,\eta}^{(t)}(v)
:=
\frac{
p_i^{(t)}(v)\exp\{\eta a_{k,t}(v)\}
}{
\sum_{z\in\mathcal V}
p_i^{(t)}(z)\exp\{\eta a_{k,t}(z)\}
}.
\label{eq:common-tilt-theory}
\end{equation}
\begin{proposition}[Minimum-KL score gain]
\label{prop:minimum-kl-local-gain}
Suppose we want to increase the linearized expected score at step $t$ by at least \(\Gamma_t\) while expending minimal KL distortion. Accordingly, the solution of
\begin{equation}
\min_{\{q_i\}}
\sum_{i\in M_t}
D_{\mathrm{KL}}\!\left(q_i\,\middle\|\,p_i^{(t)}\right)
\quad
\text{subject to}
\quad
\sum_{i\in M_t}
\left(
\mathbb E_{q_i}[a_{k,t}(V)]-\mu_{t,i}
\right)
\ge \Gamma_t
\label{eq:constrained-kl-projection}
\end{equation}
is $q_i^\star=p_{i,\eta}^{(t)} \text{ }$ for all $i\in M_t$, where \(\eta\ge0\) is chosen so that
\begin{equation}
\sum_{i\in M_t}
\left(
\mathbb E_{p_{i,\eta}^{(t)}}[a_{k,t}(V)]-\mu_{t,i}
\right)
=
\Gamma_t,
\label{eq:eta-gain-equation}
\end{equation}
whenever the constraint is feasible and active. 
\end{proposition}

\begin{proof}
The Lagrangian for \eqref{eq:constrained-kl-projection} with respect to $\eta$ is
\[
\sum_i D_{\mathrm{KL}}(q_i\|p_i^{(t)})
-
\eta
\left[
\sum_i
\left(
\mathbb E_{q_i}[a_{k,t}(V)]-\mu_{t,i}
\right)
-\Gamma_t
\right].
\]
Constants independent of \(q_i\) may be dropped, so the problem separates over positions. For each \(i\), we need to minimize
\[
D_{\mathrm{KL}}(q_i\|p_i^{(t)})
-
\eta \mathbb E_{q_i}[a_{k,t}(V)].
\]
By the Gibbs variational formula, the unique minimizer is
\[
q_i^\star(v)
=
\frac{
p_i^{(t)}(v)\exp\{\eta a_{k,t}(v)\}
}{
\sum_{z\in\mathcal V}
p_i^{(t)}(z)\exp\{\eta a_{k,t}(z)\}
}.
\]
The multiplier \(\eta\) is common because there is a single total gain constraint. If the constraint is active, \(\eta\) is chosen to satisfy
\eqref{eq:eta-gain-equation}.
\end{proof}

\begin{corollary}[Local Fisher variance]
\label{cor:local-fisher-law}
Assume there exist constants $B_t$ such that \(|a_{k,t}(v)|\le B_t\) for all \(v\in\mathcal V\), and define
\[
\sigma_{t,i}^2
:=
\operatorname{Var}_{p_i^{(t)}}(a_{k,t}(v)).
\]
Then, as \(\eta\to0\),
\begin{align}
\mathbb E_{p_{i,\eta}^{(t)}}[a_{k,t}(v)]-\mu_{t,i}
&=
\eta\sigma_{t,i}^2
+
O(B_t^3|\eta|^2),
\label{eq:local-mean-expansion-theory}
\\
D_{\mathrm{KL}}\!\left(
p_{i,\eta}^{(t)}
\,\middle\|\,
p_i^{(t)}
\right)
&=
\frac12\eta^2\sigma_{t,i}^2
+
O(B_t^3|\eta|^3).
\label{eq:local-kl-expansion-theory}
\end{align}
The constants in the \(O(\cdot)\) terms are universal.
\end{corollary}

\begin{proof}
Let \(A_{t,i}(\eta)=\log\mathbb E_{p_i^{(t)}}[\exp\{\eta a_{k,t}(v)\}]\), thus
\[
\mathbb E_{p_{i,\eta}^{(t)}}[a_{k,t}(v)]
=
A_{t,i}'(\eta),
\qquad
D_{\mathrm{KL}}\!\left(p_{i,\eta}^{(t)}\|p_i^{(t)}\right)
=
\eta A_{t,i}'(\eta)-A_{t,i}(\eta).
\]
Moreover \(A_{t,i}'(0)=\mu_{t,i}\) and \(A_{t,i}''(0)=\sigma_{t,i}^2\). Since \(a_t\) is bounded by \(B_t\), the third derivative of \(A_{t,i}\) is bounded in a neighborhood of zero by a universal constant times \(B_t^3\). Taylor expansion gives both claims.
\end{proof}

The projection in Proposition~\ref{prop:minimum-kl-local-gain} allocates expected score gain per position by
\[
\Delta_{t,i}(\eta)
:=
\mathbb E_{p_{i,\eta}^{(t)}}[a_{k,t}(v)]-\mu_{t,i}
=
A_{t,i}'(\eta)-A_{t,i}'(0).
\]
Thus the slope \(\eta\) is common across positions, while the expected gain varies through the local distribution of \(a_{k,t}(v)\). In a small bias regime, Corollary~\ref{cor:local-fisher-law} gives
\[
\Delta_{t,i}(\eta)
=
\eta\sigma_{t,i}^2+O(B_t^3\eta^2),
\]
so positions with larger local Fisher variance receive more expected detector gain.

\subsection{Soundness}
\label{sec:theory:soundness}

The following is a key-randomized soundness result on the detector in Equation \eqref{eq:detect-rule}. For any unwatermarked text \(y\), we bound the probability that the detector flags \(y\) when the secret watermark key is drawn at random.

\begin{theorem}
\label{thm:key-randomized-soundness}
Suppose $y$ is independent of the watermark key $k$. For any \(\tau>0\), if \(h_k(y)\ne0\), then
\begin{equation}
\mathbb P_{u_k}\{S_k(y)\ge \tau\mid h_k(y)\}
\le
\exp\!\left(
-
\frac{(\tau+\lambda\|h_k(y)\|_2^2)^2}
{8\|h_k(y)\|_2^2}
\right).
\label{eq:soundness-conditional}
\end{equation}
Uniformly over \(h_k\),
\begin{equation}
\mathbb P_{u_k}\{S_k(y)\ge \tau\mid h_k(y)\}
\le
\exp\!\left(-\frac{\lambda\tau}{2}\right).
\label{eq:soundness-uniform}
\end{equation}
The threshold
\begin{equation}
\tau_\alpha
=
\frac{2}{\lambda}\log\frac{1}{\alpha}
\label{eq:soundness-threshold}
\end{equation}
controls the key-randomized false-positive probability at level \(\alpha\).
\end{theorem}

\begin{proof}
Conditioning on \(h_k\), the event \(S_k(y)\ge\tau\) is
\[
\langle u_k,h_k(y)\rangle
\ge
\frac{\tau+\lambda\|h_k(y)\|_2^2}{2}.
\]
Recall that $u_k$ and $h_k$ are independent—thus the left-hand side is a centered Rademacher sum. Hoeffding's inequality gives
\[
\mathbb P\{\langle u_k,h\rangle\ge a\mid h\}
\le
\exp\!\left(-\frac{a^2}{2\|h\|_2^2}\right).
\]
Taking \(a=(\tau+\lambda\|h_k(y)\|_2^2)/2\) gives
\eqref{eq:soundness-conditional}. If \(h_k(y)=0\), the event is impossible for \(\tau>0\). For \(z=\|h_k(y)\|_2^2>0\),
\[
\frac{(\tau+\lambda z)^2}{z}
=
\frac{\tau^2}{z}+2\lambda\tau+\lambda^2z
\ge
4\lambda\tau,
\]
so \eqref{eq:soundness-conditional} implies
\[
\mathbb P_{u_k}\{S_k(y)\ge \tau\mid h_k(y)\}
\le
\exp\!\left(-\frac{4\lambda\tau}{8}\right)
=
\exp\!\left(-\frac{\lambda\tau}{2}\right).
\]
Solving \(\exp(-\lambda\tau/2)\le\alpha\) gives
\eqref{eq:soundness-threshold}.
\end{proof}

The threshold \eqref{eq:soundness-threshold} is rather conservative because it ignores the null distribution of \(h_k(Y)\). Therefore, empirical estimation of the threshold (e.g., \(\tau_N\) from Section~\ref{sec:method:detect}) is preferable. This theorem primarily shows that the norm penalty prevents arbitrarily large sketches from creating false positives.

\subsection{Edit sensitivity of the score}
\label{sec:theory:edit-robustness}

A commonly studied aspect of watermark robustness is how the detection statistic degrades under substitutions, insertions, and deletions. Here, we establish bounds for our score under arbitrary edits.

\begin{assumption}[Bounded token map]
\label{ass:bounded-phi}
There exists \(L_\phi<\infty\) such that for all $v\in\mathcal V$, $
\|\operatorname{vec}(\phi_k(v))\|_2\le L_\phi$.
\end{assumption}
For Count-Sketch, \(L_\phi=\sqrt d\), since each token activates one signed bucket in each of the \(d\) rows.

\begin{lemma}[Sketch continuity]
\label{lem:normalized-sketch-edit-lipschitz}
For any two texts \(y\in\mathcal V^N\) and \(y'\in\mathcal V^{N'}\), let their Levenshtein edit distance be \(E=\operatorname{ed}(y,y')\). If \(N_{\min}:=\min\{N,N'\}\ge1\), then
\begin{equation}
\|h_k(y)-h_k(y')\|_2
\le
\frac{3L_\phi E}{\sqrt{N_{\min}}}.
\label{eq:sketch-edit-bound-length-changing}
\end{equation}
If \(N=N'\), the stronger bound holds:
\begin{equation}
\|h_k(y)-h_k(y')\|_2
\le
\frac{2L_\phi E}{\sqrt N}
\label{eq:sketch-edit-bound-substitution}
\end{equation}
\end{lemma}

\begin{proof}
Let \(H=H_k(y)\) and \(H'=H_k(y')\). One substitution changes each unnormalized sketch by at most \(2L_\phi\), and each insertion or deletion changes it by at most \(L_\phi\). Thus, $\|H-H'\|_2\le 2L_\phi E$. Assume without loss of generality that \(N\le N'\). Then
\[
\left\|
\frac{H}{\sqrt N}
-
\frac{H'}{\sqrt{N'}}
\right\|_2
\le
\frac{\|H-H'\|_2}{\sqrt N}
+
\left|
\frac{1}{\sqrt N}
-
\frac{1}{\sqrt{N'}}
\right|
\|H'\|_2 .
\]
Since \(\|H'\|_2\le L_\phi N'\) and \(N'-N\le E\), the second term on the right-hand side is at most
\(L_\phi E/\sqrt N\). The first term is at most \(2L_\phi E/\sqrt N\). This gives
\eqref{eq:sketch-edit-bound-length-changing}. If \(N=N'\), the normalization term vanishes, giving
\eqref{eq:sketch-edit-bound-substitution}.
\end{proof}

\begin{theorem}[Worst-case score sensitivity]
\label{thm:penalized-edit-robustness}
Let \(y\in\mathcal V^N\), \(y'\in\mathcal V^{N'}\), \(h=h_k(y)\), and \(h'=h_k(y')\). Then
\begin{equation}
|S_k(y)-S_k(y')|
\le
\left(
2\sqrt D
+
\lambda(\|h\|_2+\|h'\|_2)
\right)
\|h-h'\|_2.
\label{eq:score-lipschitz-generic}
\end{equation}
Consequently, under Assumption~\ref{ass:bounded-phi},
\begin{equation}
|S_k(y)-S_k(y')|
\le
\frac{3L_\phi E}{\sqrt{N_{\min}}}
\left(
2\sqrt D
+
\lambda(\|h\|_2+\|h'\|_2)
\right)
\label{eq:penalized-edit-bound-length-changing}
\end{equation}
for arbitrary edits, and
\begin{equation}
|S_k(y)-S_k(y')|
\le
\frac{2L_\phi E}{\sqrt N}
\left(
2\sqrt D
+
\lambda(\|h\|_2+\|h'\|_2)
\right)
\label{eq:penalized-edit-bound-substitution}
\end{equation}
for \(E\) substitutions at fixed length \(N\).
\end{theorem}

\begin{proof}
Let \(\Delta=h-h'\). Then
\[
S_k(y)-S_k(y')
=
2\langle u_k,\Delta\rangle
-
\lambda(\|h\|_2^2-\|h'\|_2^2).
\]
Since \(\|u_k\|_2=\sqrt D\) and
\[
|\|h\|_2^2-\|h'\|_2^2|
\le
(\|h\|_2+\|h'\|_2)\|\Delta\|_2,
\]
we obtain \eqref{eq:score-lipschitz-generic}. Applying Lemma~\ref{lem:normalized-sketch-edit-lipschitz} gives \eqref{eq:penalized-edit-bound-length-changing} and \eqref{eq:penalized-edit-bound-substitution}.
\end{proof}

\bibliographystyle{plainnat}
\bibliography{global_sketch_wm}

\end{document}